\documentclass[letterpaper, 10pt, conference]{IEEEconf}

\usepackage{amssymb,bm,epstopdf,graphicx,float,caption,bbm,amsmath,breqn,multirow,url,cite,comment}

\usepackage{amsthm}
\newtheorem{theorem}{Theorem}

\theoremstyle{definition}
\newtheorem{defn}{Definition}

\begin{document}

\title{\LARGE \bf The Strategic LQG System: A Dynamic Stochastic VCG Framework for Optimal Coordination}

\author{Ke Ma and P. R. Kumar
\thanks{This work is supported in part by NSF Contract ECCS-1760554, 
NSF Science \& Technology Center Grant CCF-0939370, the Power Systems Engineering Research Center (PSERC), NSF Contract IIS-1636772, and NSF Contract 
ECCS-1760554.}
\thanks{Ke Ma and P. R. Kumar are with the Department of Electrical \& Computer Engineering, Texas A\&M University, College Station, TX 77843, USA
{\tt\small ke.ma@tamu.edu, prk.tamu@gmail.com}}
}

\maketitle
\thispagestyle{empty}
\pagestyle{empty}

\begin{abstract}
The classic Vickrey-Clarke-Groves (VCG) mechanism ensures incentive compatibility, i.e., that truth-telling of all agents is a dominant strategy, for a static one-shot game. However, in a dynamic environment that unfolds over time, the agents' intertemporal payoffs depend on the expected future controls and payments, and a direct extension of the VCG mechanism is not sufficient to guarantee incentive compatibility. In fact, it does not appear to be feasible to construct mechanisms that ensure the dominance of dynamic truth-telling for agents comprised of general stochastic dynamic systems. The contribution of this paper is to show that such a dynamic stochastic extension does exist for the special case of Linear-Quadratic-Gaussian (LQG) agents with a careful construction of a sequence of layered payments over time.

We propose a layered version of a modified VCG mechanism for payments that decouples the intertemporal effect of current bids on future payoffs, and prove that truth-telling of dynamic states forms a dominant strategy if system parameters are known and agents are rational.

An important example of a problem needing such optimal dynamic coordination of stochastic agents arises in power systems where an Independent System Operator (ISO) has to ensure balance of generation and consumption at all time instants, while ensuring social optimality (maximization of the sum of the utilities of all agents). Addressing strategic behavior is critical as the price-taking assumption on market participants may not hold in an electricity market. Agents, can lie or otherwise game the bidding system. The challenge is to determine a bidding scheme between all agents and the ISO that maximizes social welfare, while taking into account the stochastic dynamic models of agents, since renewable energy resources such as solar/wind are stochastic and dynamic in nature, as are consumptions by loads which are influenced by factors such as local temperatures and thermal inertias of facilities.
\end{abstract}

\section{introduction}
Mechanism design is the sub-field of game theory that considers how to implement socially optimal solutions to problems involving multiple self-interested agents, each with a private utility function. A typical approach in mechanism design is to provide financial incentives such as payments to promote truth-telling of utility function parameters from agents. Consider for example the Independent System Operator (ISO) problem of electric power systems in which the ISO aims to maximize social welfare and maintain balance of generation and consumption while each generator/load has a private utility function. The classic Vickery-Clarke-Groves (VCG) mechanism \cite{JOFI:JOFI2789} has played a central role in classic mechanism design since it ensures incentive compatibility, i.e., truth-telling of utility functions of all agents forms a dominant strategy, and social efficiency, i.e., the sum of utilities of all agents is maximized. Indeed, the outcome generated by the VCG mechanism is stronger than a Nash equilibrium in the sense that it is \emph{strategy-proof}, meaning that truth-telling of utility functions is optimal irrespective of what others are bidding. In fact, Green, Laffont and Holmstrom \cite{Green1979} show that VCG mechanisms are the \emph{only} mechanisms that are both efficient and strategy-proof if payoffs are quasi-linear.

While the VCG mechanism is applicable to a static one-shot game, it does not work for stochastic dynamic games. In a dynamic environment that unfolds over time, the agents' intertemporal payoffs depend on the expected future controls and payments, and a direct extension of the VCG mechanism is not sufficient to guarantee incentive compatibility. A fundamental difference between dynamic and static mechanism design is that in the former, an agent can bid an untruthful utility function conditional on his past bids (which need not be truthful) and past allocations (from which he can make an inference about other agents' utility functions). Here we should note that for dynamic deterministic systems, by collecting the VCG payments as a lump sum of all the payments over the entire time horizon at the beginning, incentive compatibility is still assured. However, for a dynamic \emph{stochastic} system, the states are private random variables and there is no incentive for agents to bid their states truthfully if VCG payments are collected in the same way as for dynamic deterministic systems. In fact, it does not appear to be feasible to construct mechanisms that ensure the dominance of dynamic truth-telling for agents comprised of general stochastic dynamic systems. 

Nevertheless, for the special case of Linear-Quadratic-Gaussian (LQG) agents, where agents have linear state equations, quadratic utility functions and additive white Gaussian noise, we show in this paper that a dynamic stochastic extension of the VCG mechanism does exist, based on a careful construction of a sequence of layered payments over time. For a set of LQG agents, we propose a modified layered version of the VCG mechanism for payments that decouples the intertemporal effect of current bids on future payoffs, and prove that truth-telling of dynamic states forms a dominant strategy if system parameters are known and agents are rational. ``Rational" means that an agent will adopt a dominant strategy if it is the unique one, and it will act on the basis that it and others will do so at future times.

An important example of a problem needing such optimal dynamic coordination of stochastic agents arises in the ISO problem of power systems. In general, agents may have different approaches to responding to the prices set by the ISO. If each agent acts as a \emph{price taker}, i.e., it honestly discloses its energy consumption at the announced prices, a \emph{competitive equilibrium} would be reached among agents. However, when each agent becomes a \emph{price anticipator}, and it is critical for the ISO to design a market mechanism that is strategy-proof (i.e., incentive compatible). The challenge for the ISO is to determine a bidding scheme between all agents (producers and consumers) and the ISO that maximizes social welfare, while taking into account the stochastic dynamic models of agents, since renewable energy resources such as solar/wind are stochastic and dynamic in nature, as are consumptions by loads which are influenced by factors such as local temperatures and thermal inertias of facilities. Currently, the ISO solicits bids from generators and Load Serving Entities (LSEs) and operates two markets: a day-ahead market and a real-time market. The day-ahead market lets market participants commit to buy or sell wholesale
electricity one day before the operating day, to satisfy energy demand bids and
to ensure adequate scheduling of resources to meet the next
day's anticipated load. The real-time market lets market
participants buy and sell wholesale electricity during the
course of the operating day to balance the differences between
day-ahead commitments and the actual real-time demand and
production \cite{isone}. Our layered VCG mechanism fits perfectly in the real-time market, as we will see in the sequel.

The rest of the paper is organized as follows. In Section \ref{RWCDC}, a survey of related works is presented. This is followed by a complete description of the classic VCG framework for the static and dynamic deterministic problem in Section \ref{SDCDC}. A layered VCG payment scheme is introduced for the dynamic stochastic problem in Section \ref{SCDC}. Section \ref{CRCDC} concludes the paper.

\section{Related Works}\label{RWCDC}

In recent years, several papers have been written with the aim of exploring issues arising in dynamic mechanism design. In order to achieve ex post incentive compatibility, Bergemann and Valimaki \cite{RePEc:cwl:cwldpp:1616} propose a generalization of the VCG mechanism based on the marginal contribution of each agent and show that ex post participation constraints are satisfied under some conditions. Athey and Segal \cite{ECTA:ECTA1375} consider a similar model and focus on budget balance of the mechanism. Pavan et al. \cite{Pavan2009Dynamic} derives first-order conditions under which incentive compatibility is guaranteed by generalizing Mirrlees's \cite{10.2307/2296779} envelope formula of static mechanisms. Cavallo et al. \cite{DBLP:journals/corr/CavalloPS12} considers a dynamic Markovian model and derives a sequence of Groves-like payments which achieves Markov perfect equilibrium. Bapna and Weber \cite{Bapna2005EfficientDA} solves a sequential allocation problem by formulating it as a multi-armed bandit problem. Parkes and Singh \cite{NIPS2003_2432} and Friedman and Parkes \cite{Friedman:2003:PWS:779928.779978} consider an environment with randomly arriving and departing agents and propose a ``delayed'' VCG mechanism to guarantee interim incentive compatibility. Besanko et al. \cite{BESANKO198533} and Battaglini et al. \cite{RePEc:pri:metric:wp046_2012_battaglini_lamba_optm_dyn_contract_10october2012_short.pdf} characterize the optimal infinite-horizon mechanism for an agent modeled as a Markov process, with Besanko considering a linear AR(1) process over a continuum of states, and Battaglini focusing on a two-state Markov chain. Farhadi et al. \cite{10.1007/978-3-319-67540-4_8} propose a dynamic mechanism that is incentive compatible, individual rational, ex-ante budget balance and social efficient based on the set of inference signals. However, their notion of incentive compatibility is in a weaker Nash sense, i.e., given other agents report truthfully, agent $i$'s best reaction is to report truthfully. Our dynamic VCG mechanism on the other hand, guarantees incentive compatibility in weakly dominant strategies, i.e., irrespective of what other agents are bidding, agent $i$'s best strategy is to report truthfully. Bergemann and Pavan \cite{BERGEMANN2015679} have an excellent survey on recent research in dynamic mechanism design and a more recent survey paper by Bergemann and Valimaki \cite{RePEc:cwl:cwldpp:2102} further discusses dynamic mechanism design problem with risk-averse agents and the relationship between dynamic mechanism and optimal contracts. 

To our knowledge, there does not appear to be any result that ensures dominance of dynamic truth-telling for agents comprised of LQG systems.
\section{The Static and Dynamic Deterministic VCG}\label{SDCDC}
Let us begin by considering the simpler static deterministic case. Suppose there are $N$ agents, with each agent having a utility function $F_{i}(u_{i})$, where $u_{i}$ is the amount of energy produced/consumed by agent $i$. $F_{i}(u_{i})$ depends only on its own consumption/generation $u_{i}$. However, for convenience of notation, we will occasionally abuse notation and write $F_{i}(u)$ with the implicit understanding that it only depends on the $i$-th component $u_{i}$ of $u$.

Let $\boldsymbol{u}:=(u_{1},...,u_{N})^{T}$, $\boldsymbol{u}_{-i}:=(u_{1},...,u_{i-1},u_{i+1},...,u_{N})^{T}$, and let $F := (F_{1}, \ldots , F_{n})$.  

In the VCG mechanism, each agent is asked to bid its utility function $\hat{F}_{i}$. The agent can lie, so $\hat{F}_{i}$ may not be equal to $ F_{i}$. (As for $F$, we denote $\hat{F} := (\hat{F}_{1}, \ldots \hat{F}_{n})$). After obtaining the bids, the ISO calculates $\boldsymbol{u^{*}}(\hat{F})$ as the optimal solution to the following problem:
\begin{equation*}
\max_{\boldsymbol{u}}\sum_{i}\hat{F}_{i}(u_{i})
\end{equation*}
\noindent subject to

\begin{equation*}
\sum_{i}u_{i}=0.
\end{equation*}
The last equality ensures balance between generation and consumption. Each agent is then assigned to produce/consume $u_{i}^{*}(\hat{F})$, and is obliged to do so, accruing a utility $F_{i}\left(u_{i}^{*}\left(\hat{F}\right)\right)$. Following the rules that it has announced a priori before receiving the bids, the ISO then collects a payment $p_{i}(\hat{F})$ from agent $i$, defined as follows:
\begin{equation*}
p_{i}(\hat{F}):=\sum_{j\ne i}\hat{F}_{j}(\boldsymbol{u}^{(i)})-\sum_{j\ne i}\hat{F}_{j}(\boldsymbol{u^{*}}),
\end{equation*} 
\noindent where $\boldsymbol{u}^{(i)}$ is defined as the optimal solution to the following problem:
\begin{equation*}
\max_{\boldsymbol{u}_{-i}}\sum_{j\ne i}\hat{F}_{j}(u_{j})
\end{equation*}
\noindent subject to
\begin{equation*}
\sum_{j\ne i}u_{j}=0.
\end{equation*}
We can see that $p_{i}$ is the cost to the rest of the agents due to agent $i$'s presence, which leads agents to internalize the social externality.

In fact, the VCG mechanism is a special case of the Groves mechanism \cite{10.2307/1914085}, where payment $p_{i}$ is defined as:
\begin{equation*}
p_{i}(\hat{F})=h_{i}(\boldsymbol{\hat{F}}_{-i})-\sum_{j\ne i}\hat{F}_{j}(\boldsymbol{u}^{*}(\hat{F})).
\end{equation*}
where $h_{i}$ is any arbitrary function and $\hat{F}_{-i}:=(\hat{F}_{1},..,\hat{F}_{i-1},\hat{F}_{i+1},...,\hat{F}_{N})$. Truth-telling is a dominant strategy in the Groves mechanism \cite{10.2307/1914085}. That is, regardless of other agents' strategies, an agent cannot do better than truthfully declaring its utility function.

\begin{theorem}\cite{10.2307/1914085}\label{vcgic}
	Truth-telling ($\hat{F}_{i}
	\equiv F_{i}$) is the dominant strategy equilibrium in Groves mechanism.
\end{theorem}
\begin{proof}
	Suppose agent $i$ announces the true utility function $F_{i}$. Let
	$\bar{F}:=(\hat{F}_{1},...\hat{F}_{i-1},F_{i},\hat{F}_{i+1},...,\hat{F}_{N})$ and $\bar{F}_{-i}:=(\hat{F}_{1},...\hat{F}_{i-1},\hat{F}_{i+1},...,\hat{F}_{N})$. Let $\bar{F}(u):=\sum_{i}\bar{F}_{i}(u_{i})$.  Let $\bar{u}^{*}_{i}$ be what ISO assigns, and $p_{i}(\bar{F})$ be what ISO charges, when $\bar{F}$ is announced by the agents.  Let $u^{*}_{i}$ be what ISO assigns and $p_{i}(\hat{F})$ be what ISO charges when $\hat{F}$ is announced by agents.
	
	Note that $\bar{F}_{-i} = \hat{F}_{-i}$, and so
	$h_i(\bar{F}_{-i}) = h_i(\hat{F}_{-i})$. Hence for agent $i$, the difference between the net utilities resulting from announcing $F_{i}$ and $\hat{F}_{i}$ is
	\begin{align*}
	&\Big[F_{i}(\bar{u}^{*}_{i})-p_{i}(\bar{F})\Big]-\Big[F_{i}(u^{*}_{i})
	-p_{i}(\hat{F})\Big]\\
	=&F_{i}(\bar{u}^{*}_{i})-h_{i}(\bar{F}_{-i})+\sum_{j\ne i}\hat{F}_{j}(\bar{u}^{*}_{i})
	-F_{i}(u^{*}_{i})+h_{i}(\hat{F}_{-i})\\&-\sum_{j\ne i}\hat{F}_{j}(u^{*}_{i})=\bar{F}(\bar{u}^{*})-\bar{F}(u^{*})\ge 0,
	\end{align*}
	where the last inequality holds since $\bar{u}^{*}$ is the optimal solution to the social welfare problem with utility functions $\bar{F}$.
\end{proof}

The above VCG scheme can be extended to the important case of dynamic systems. We first consider the deterministic case. This is a straightforward extension of the static case since one can consider the sequence of actions taken by an agent as a vector action. That is, one can simply view the problem as an open-loop control problem, where the entire decision on the sequence of controls to be employed is taken at the initial time, and so treatable as a static problem.

 For agent $i$, let $F_{i,t}(x_{i}(t),u_{i}(t))$ be the one-step utility function at time $t$. Suppose that the state of agent $i$ evolves as:
\begin{equation*}
x_{i}(t+1)=g_{i,t}(x_{i}(t),u_{i}(t)).
\end{equation*}
The ISO asks each agent $i$ to bid its one-step utility functions $\{\hat{F}_{i,t}(x_{i}(t),u_{i}(t)), t = 0, 1, \ldots , T-1\}$, state equation $\{\hat{g}_{i,t}, t = 0, 1, \ldots , T-1\}$, and initial condition $\hat{x}_{i,0}$. The ISO then calculates $(x^{*}_{i}(t),u^{*}_{i}(t))$ as the optimal solution, assumed to be unique, to the following utility maximization problem:
\begin{equation*}
\max \ \sum_{i=1}^{N}\sum_{t=0}^{T-1}\hat{F}_{i}(x_{i}(t),u_{i}(t))
\end{equation*}
\noindent subject to
\begin{equation*}
x_{i}(t+1)=\hat{g}_{i}(x_{i}(t),u_{i}(t)),\text{ for }\forall i \text{ and } \forall t,
\end{equation*}
\begin{equation*}
\sum_{i=1}^{N}u_{i}(t)=0,\text{ for }\forall t,
\end{equation*}
\begin{equation*}
x_{i}(0)=\hat{x}_{i,0},\text{ for }\forall i.
\end{equation*}
We denote this problem as $(\hat{F},\hat{g},\hat{x}_{0})$. We can extend the notion of VCG payment $p_{i}$ to the deterministic dynamic system as follow. Let
\begin{equation*}
p_{i}:=\sum_{j\ne i}\sum_{t=0}^{T-1}\hat{F}_{j}(x^{(i)}_{j}(t),u^{(i)}_{j}(t))-\sum_{j\ne i}\sum_{t=0}^{T-1}\hat{F}_{j}(x^{*}_{j}(t),u^{*}_{j}(t)).
\end{equation*}
Here $(x^{(i)}_{i}(t),u^{(i)}_{i}(t))$ is the optimal solution to the following problem, which is assumed to be unique:
\begin{equation*}
\max \ \sum_{j\ne i}\sum_{t=0}^{T-1}\hat{F}_{j}(x_{j}(t),u_{j}(t))
\end{equation*}
\noindent subject to
\begin{equation*}
x_{j}(t+1)=\hat{g}_{j}(x_{j}(t),u_{j}(t)),\text{ for }j\ne i \text{ and } \forall t,
\end{equation*}
\begin{equation*}
\sum_{j\ne i}u_{j}(t)=0,\text{ for }\forall t,
\end{equation*}
\begin{equation*}
x_{j}(0)=\hat{x}_{j,0},\text{ for }j\ne i.
\end{equation*}
More generally, we can consider a Groves payment $p_{i}$ defined as:
\begin{equation*}
p_{i}:=h_{i,t}(\boldsymbol{\hat{F}}_{-i})-\sum_{j\ne i}\sum_{t=0}^{T-1}\hat{F}_{j}(x^{*}_{j}(t),u^{*}_{j}(t)).
\end{equation*}
where $h_{i,t}$ is any arbitrary function. We first show in the following theorem that truth-telling is still the dominant strategy equilibrium in Groves mechanism.
\begin{theorem}\label{dtruth}
Truth-telling of utility function, state dynamics and initial condition ($\hat{F}_{i}=F_{i}$, $\hat{g}_{i}=g_{i}$ and $\hat{x}_{i,0}=x_{i,0}$) is a dominant strategy equilibrium under the Groves mechanism for a dynamic system.
\end{theorem}
\begin{proof}
	Let $\hat{F}:=(\hat{F}_{1},...,\hat{F}_{i},...,\hat{F}_{N})$, $\hat{g}:=(\hat{g}_{1}...,\hat{g}_{i},...,\hat{g}_{N})$, and $\hat{x}_{0}:=(\hat{x}_{1,0},...,\hat{x}_{i,0},...,\hat{x}_{N,0})$.
	Suppose agent $i$ announces the true one-step utility function $F_{i}$, true state dynamics $g_{i}$, and true initial condition $x_{i,0}$. Let
	 $\bar{F}:=(\hat{F}_{1},...\hat{F}_{i-1},F_{i},\hat{F}_{i+1},...,\hat{F}_{N})$, $\bar{g}:=(\hat{g}_{1},...\hat{g}_{i-1},g_{i},\hat{g}_{i+1},...,\hat{g}_{N})$, and $\bar{x}_{0}:=(\hat{x}_{1,0},...\hat{x}_{i-1,0},x_{i,0},\hat{x}_{i+1,0},...,\hat{x}_{N,0})$. Let $(\bar{x}^{*}_{i}(t),\bar{u}^{*}_{i}(t))$ be what ISO assigns and $p_{i}(\bar{F},\bar{g},\bar{x}_{0})$ be what ISO charges when $(\bar{F},\bar{g},\bar{x_{0}})$ is announced by agents. Let $(x^{*}_{i}(t),u^{*}_{i}(t))$ be what ISO assigns and $p_{i}(\hat{F},\hat{g},\hat{x}_{0})$ be what ISO charges when $(\hat{F},\hat{g},\hat{x}_{0})$ is announced by agents. Let $\bar{F}(x_{i}(t),u_{i}(t)):=\sum_{i}\bar{F}_{i}(x_{i}(t),u_{i}(t))$.
	 
	 For agent $i$, the difference between net utility resulting from announcing $(F_{i},g_{i},x_{i,0})$ and $(\hat{F}_{i},\hat{g}_{i},\hat{x}_{i,0})$ is
	 \begin{align*}
	 \Big[\sum_{t}F_{i}(\bar{x}^{*}_{i}(t),\bar{u}^{*}_{i}(t))-p_{i}(\bar{F},\bar{g},\bar{x}_{0})\Big]-\Big[\sum_{t}&F_{i}(x^{*}_{i}(t),u^{*}_{i}(t))\\
	 &-p_{i}(\hat{F},\hat{g},\hat{x}_{0})\Big]
	 \end{align*}
	 \begin{align*}
	 &=\sum_{t}F_{i}(\bar{x}^{*}_{i}(t),\bar{u}^{*}_{i}(t))-h_{i,t}(\bar{F}_{-i})+\sum_{j\ne i}\sum_{t}\hat{F}_{j}(\bar{x}^{*}_{i}(t),\bar{u}^{*}_{i}(t))
	 \\
	 &-\sum_{t}F_{i}(x^{*}_{i}(t),u^{*}_{i}(t))+h_{i,t}(\hat{F}_{-i})-\sum_{j\ne i}\sum_{t}\hat{F}_{j}(x^{*}_{i}(t),u^{*}_{i}(t))\\
	 &=\sum_{t}\bar{F}(\bar{x}^{*}_{i}(t),\bar{u}^{*}_{i}(t))-\sum_{t}\bar{F}(x^{*}_{i}(t),u^{*}_{i}(t))\ge 0,
	 \end{align*}
	 since $(\bar{x}^{*}_{i}(t),\bar{u}^{*}_{i}(t))$ is the optimal solution to the problem $(\bar{F},\bar{g},\bar{x}_{0})$.
\end{proof}

\section{The Dynamic Stochastic VCG}\label{SCDC}
In the above section, we have shown that the VCG mechanism can be naturally extended to dynamic deterministic systems by employing an open-loop solution. However, when agents are dynamic stochastic systems, we need to consider closed-loop solutions. Such closed-loop controls depend on the observations of the agents, which are generally private. So the states of the system are private random variables. Hence the problem becomes one of additionally ensuring that each agents reveals its ``true" states at all times. This additional complication appears to prevent a solution for general systems. However, as we will see, in the case of LQG agents one can indeed ensure the dominance of truth telling strategies that reveal the true states. However, it does not appear feasible to also then ensure that the agents reveal their true state equations and cost functions.

To obtain the correct payment structure, we will need to carefully redefine the VCG payments such that incentive compatibility is still assured for the special case of linear quadratic Gaussian (LQG) systems. As noted above, one cannot treat the system as an open-loop system as in the previous section. In particular, this necessitates collecting payments from agents at each step. For agent $i$, let $w_{i}\sim\mathcal{N}(0,\sigma_{i})$ be the discrete-time additive Gaussian white noise process affecting state $x_{i}(t)$ via:
	\begin{equation*}
	x_{i}(t+1)=a_{i}x_{i}(t)+b_{i}u_{i}(t)+w_{i}(t),
	\end{equation*}
	where $x_{i}(0)\sim\mathcal{N}(0,\zeta_{i})$ and is independent of $w_{i}$.
	Each agent has a one-step utility function
	\begin{equation*}
	F_{i}(x_{i}(t),u_{i}(t))=q_{i}x_{i}^{2}(t)+r_{i}u_{i}^{2}(t).
	\end{equation*} 
	Let $X(t)=[x_{1}(t),...,x_{N}(t)]^{T}$, $U(t)=[u_{1}(t),...,u_{N}(t)]^{T}$ and $W(t)=[w_{1}(t),...,w_{N}(t)]^{T}$. Let $Q=diag(q_{1},...,q_{N})\leq 0$, $R=diag(r_{1},...,r_{N})<0$, $A=diag(a_{1},...,a_{N})$, $B=diag(b_{1},...,b_{N})$,
	$\Sigma=diag(\sigma_{1},...,\sigma_{N})>0$ and $Z=diag(\zeta_{1},...,\zeta_{N})>0$. Let $RSW:=\sum_{i=1}^{N}\sum_{t=0}^{T-1}[X^{T}(t)QX(t)+U^{T}(t)RU(t)]$ be the random variable denoting the social welfare of the agents, and let $SW:=\mathbb{E}[RSW]$ denote the expected social welfare. The random social welfare could also be called the ``ex-post social welfare''.
	The ISO aims to maximize the social welfare, leading to the following LQG problem:
	\begin{equation*}
	\max \ \mathbb{E}\sum_{i=1}^{N}\sum_{t=0}^{T-1}\left[X^{T}(t)QX(t)+U^{T}(t)RU(t)\right]
	\end{equation*}
	\noindent subject to
	\begin{equation*}
	X(t+1)=AX(t)+BU(t)+W(t),
	\end{equation*}
	\begin{equation}\label{ubalance}
	1^{T}U(t)=0,\text{ for }\forall t,
	\end{equation}
	\begin{equation*}
	X(0)\sim\mathcal{N}(0,Z), W\sim\mathcal{N}(0,\Sigma).
	\end{equation*}
	We will rewrite the random social welfare and thereby the social welfare in terms more convenient for us. We will decompose $X(t)$ as:
	\begin{equation}\label{xdecompose}
	X(t):=\sum_{s=0}^{t}X(s,t),\ 0\le t\le T-1,
	\end{equation}
	where $X(s,s) := W(s-1)$ for $s\geq 1$ and $X(0,0):=X(0)$. Let
	\begin{equation}\label{projection}
	X(s,t) := AX(s,t-1)+BU(s,t-1), \text{ } 0 \leq s \leq t-1,
	\end{equation}
	with $U(s,t)$ yet to be specified. We suppose that $U(t)$ can also be decomposed as:
	\begin{equation}\label{udecompose}
	U(t):=\sum_{s=0}^{t}U(s,t),\ 0\le t \le T-1.
	\end{equation}
	Then regardless of how the $U(s,t)$'s are chosen, as long as the $U(s,t)$'s for $0 \leq s \leq t$ are indeed a decomposition of $U(t)$, i.e., \eqref{udecompose} is satisfied, the random social welfare can be written in terms of $X(s,t)$'s and $U(s,t)$'s as:	
	\begin{equation*}
	RSW=\sum_{s=0}^{T-1}L_{s},
	\end{equation*}
	where $L_{s}$ for $s\ge 1$ is defined as:
	\begin{align}\label{ls}
	L_{s}:&=\sum_{t=s}^{T-1}\bigg[X^{T}(s,t)QX(s,t)+U^{T}(s,t)RU(s,t)\\\nonumber+&2\left(\sum_{\tau=0}^{s-1}X(\tau,t)\right)QX(s,t)+2\left(\sum_{\tau=0}^{s-1}U(\tau,t)\right)RU(s,t)\bigg],
	\end{align}
	and $L_{0}$ is defined as:
	\begin{equation*}
	L_{0}:=\sum_{t=0}^{T-1}\Big[X^{T}(0,t)QX(0,t)+U^{T}(0,t)RU(0,t)\Big].
	\end{equation*}
	Hence,
	\begin{equation*}
	SW=\mathbb{E}\sum_{s=0}^{T-1}L_{s}.
	\end{equation*}
	
	In the scheme to follow the ISO will choose all $U(s,t)$'s for different $t$'s at time $s$ based on the information it has at time $s$. (Note that $t \geq s$). Hence $X(s,t)$ is completely determined by $W(s-1)$, and $U(s,t)$ for $s \leq t \leq T-1$. Indeed $X(s,t)$ can be regarded as the contribution to $X(t)$ of these variables.
	
	Here we assume that the ISO knows the true system parameters $Q$, $R$, $A$ and $B$. This may hold if the ISO has previously run the VCG bidding scheme for a dynamic deterministic system, or equivalently, a day-ahead market, and system parameters remain unchanged when agents participate in the real-time stochastic market. 
	
	We will consider a scheme where at each stage, the ISO asks the agents to bid their $x_{i}(s,s)$ (which is equal to $w_{i}(s-1)$) at each time $s$, for $0 \leq s \leq T-1$. Let $\hat{x}_{i}(s,s)$ be what the agents actually bid, since they may not tell the truth. Based on their bids $\hat{x}_i(s,s)$ for $1 \leq i \leq N$, the ISO solves the following problem:	
	\begin{equation*}
	\max \ L_{s}
	\end{equation*} 
	\noindent subject to
	\begin{equation*}
	1^{T}U(s,t)=0,\text{ for } s\le t \le T-1,
	\end{equation*}
	\begin{equation*}
	\hat{X}(s,s)=[\hat{x}_{1}(s,s),...,\hat{x}_{N}(s,s)]^{T}.
	\end{equation*}
	The variables $\hat{X}(s,t)$ for $t>s$ are defined as $\hat{X}(s,t)=A\hat{X}(s,t-1)+BU(s,t-1)$, that is, with zero noise in the state variable updates starting from the ``initial condition'' $\hat{X}(s,s)$.
	
	The interpretation is the following. Based on the bids, $\hat{X}(s,s)$, which is supposedly a bid of $W(s-1)$, the ISO calculates the trajectory of the linear systems from time $s$ onward, assuming zero noise from that point on. It then allocates consumptions/generations $U(s,t)$ for future periods $t$ for the corresponding deterministic linear system, with balance of consumption and production \eqref{ubalance} at each time $t$. These can be regarded as taking into account the consequences of the disturbance occurring at time $s$. More specifically, $X(s,t)$ is the trajectory resulting from the disturbance $W(s-1)$ at time $s$, and $U(s,t)$ is the adjustment made at time $s$ to allocation at time $t$ due to disturbance at time $s$.
	
	Next, the ISO collects a payment $p_{i}(s)$ from agent $i$ at time $s$ as:
	\begin{align*}
	p_{i}(s):=h_{i}(\hat{X}_{-i}(s,s))-\sum_{j\ne i}\sum_{t=s}^{T-1}\bigg[q_{j}\hat{x}_{j}^{2}(s,t)+r_{j}u_{j}^{*2}(s,t)\\+2q_{j}\left(\sum_{\tau=0}^{s-1}\hat{x}_{j}(\tau,t)\right)\hat{x}_{j}(s,t)+2r_{j}\left(\sum_{\tau=0}^{s-1}u_{j}(\tau,t)\right)u_{j}^{*}(s,t)\bigg],
	\end{align*}
	where $\hat{X}_{-i}(s,s)=[\hat{x}_{1}(s,s),...,\hat{x}_{i-1}(s,s),\hat{x}_{i+1}(s,s),...\\
	,\hat{x}_{N}(s,s)]^{T}$, and $h_{i}$ is any arbitrary function (as in the Groves mechanism).
	
	Before we prove incentive compatibility, we need to define what is meant by the term ``rational agents''. 
	\begin{defn}
	Rational Agents: We say agent $i$ is rational at time $T-1$, if it adopts a dominant strategy, when there is a unique dominant strategy. An agent $i$ is rational at time $t$ if it adopts a dominant strategy at time $t$ under the assumption that all agents including itself are rational at times $t+1, t+2, ..., T-1$, when there is a unique such dominant strategy.
	\end{defn}
	
	Rationality is defined in a recursion fashion.
\begin{theorem}
		Truth-telling of state $\hat{x}_{i}(s,s)$ for $0\le s \le T-1$, i.e., bidding $\hat{x}_i(s,s) = w_i(s-1)$, is the unique dominant strategy for the stochastic ISO mechanism, if system parameters $Q\le 0$, $R<0$, $A$ and $B$ are truthfully known, and agents are rational.
	\end{theorem}

\begin{proof}
We show by backward induction. When agent $i$ is at the last stage $T-1$, it is easy to verify that truth-telling of state (noise) is dominant, i.e., $\hat{x}_{i}(T-1,T-1)=x_{i}(T-1,T-1)$.
We next employ induction and so assume that truth-telling of states is a dominant strategy equilibrium at time $k$. If agents are rational, we can take expectation over $X(s,s)$, $s\ge k$, and since optimal feedback gain does not change with respect to time, the cross terms cancel and agent $i$'s objective aligns with the ISO's. We conclude that truth-telling $\hat{x}_{i}(k-1,k-1)=x_{i}(k-1,k-1)$ is the dominant strategy for agent $i$ at time $k-1$.
\end{proof}

\section{Concluding Remarks}\label{CRCDC}

It remains an open problem how to construct a mechanism that ensures the dominance of dynamic truth-telling for agents comprised of general stochastic dynamic systems. For the special case of LQG agents, by careful construction of a sequence of layered VCG payments over time, the intertemporal effect of current bids on future payoffs can be decoupled, and truth-telling of dynamic states is guaranteed if system parameters are known and agents are rational. Our results can be generalized to LQG systems with partial state observation and time-varying cost and/or state dynamics.

\section*{ACKNOWLEDGMENT}
The authors would like to thank Dr. Le Xie for his valuable comments and suggestions.

\bibliographystyle{IEEEtran}
\bibliography{myReference}
\end{document}